\documentclass[11pt,english]{article}
\usepackage[T1]{fontenc}
\usepackage[latin9]{inputenc}
\usepackage{xcolor}
\usepackage{pdfcolmk}
\usepackage{babel}
\usepackage{units}
\usepackage{amsmath}
\usepackage{amsthm}
\usepackage{amssymb}
\PassOptionsToPackage{normalem}{ulem}
\usepackage{ulem}
\usepackage[unicode=true,pdfusetitle,
 bookmarks=true,bookmarksnumbered=false,bookmarksopen=false,
 breaklinks=false,pdfborder={0 0 1},backref=false,colorlinks=true]
 {hyperref}
\hypersetup{linkcolor=blue,citecolor=blue,urlcolor=blue,filecolor=blue}

\makeatletter

\theoremstyle{plain}
\newtheorem{thm}{\protect\theoremname}
\theoremstyle{remark}
\newtheorem{rem}[thm]{\protect\remarkname}
\theoremstyle{plain}
\newtheorem{lem}[thm]{\protect\lemmaname}
\newtheorem{defn}[thm]{\protect\definitionname}

\usepackage{pdfsync,babel}

\makeatother

\providecommand{\lemmaname}{Lemma}
\providecommand{\remarkname}{Remark}
\providecommand{\theoremname}{Theorem}
\providecommand{\definitionname}{Definition}

\begin{document}

\title{Stochastic Quantization for the Edwards Measure of Fractional Brownian Motion with $Hd=1$.}

\author{\textbf{Wolfgang Bock}\\
 Technische Universit{\"a}t Kaiserslautern,\\
 Fachbereich Mathematik, Postfach 3049,\\
 67653 Kaiserslautern, Germany\\
 Email: bock@mathematik.uni-kl.de\and \textbf{Torben Fattler}\\
 Technische Universit{\"a}t Kaiserslautern,\\
 Fachbereich Mathematik, Postfach 3049,\\
 67653 Kaiserslautern, Germany\\
 Email: fattler@mathematik.uni-kl.de\and\textbf{ Jos{\'e} Lu{\'\i}s da Silva}\\
 CIMA, University of Madeira, Campus da Penteada,\\
 9020-105 Funchal, Portugal\\
 Email: luis@uma.pt \and \textbf{Ludwig Streit} \\
 BiBoS, Universit{\"a}t Bielefeld, Germany,\\
 CIMA, Unversidade da Madeira, Funchal, Portugal\\
 Email: streit@uma.pt}
\maketitle
\begin{abstract}
In this paper we construct a Markov process which has as invariant measure the fractional
Edwards measure based on a $d$-dimensional fractional Brownian motion, with Hurst index $H$ in the case of $Hd=1$. We use the theory of classical
Dirichlet forms. However since the corresponding self-intersection
local time of fractional Brownian motion is not Meyer-Watanabe differentiable in this case, we show the closability of the form via quasi translation
invariance of the fractional Edwards measure along shifts in the corresponding
fractional Cameron-Martin space. 
\end{abstract}

\tableofcontents

\section{Introduction}

In its original form the Edwards model was a proposal to modify the Wiener
measure $\mu _{0}$ for d-dimensional Brownian motion by a factor which would
exponentially suppress self-intersections of sample paths. Informally

\begin{equation*}
d\mu _{g}=Z^{-1}e^{-gL}d\mu _{0},
\end{equation*}%
where $L$ is the self-intersection local time of Brownian motion, see e.g. \cite{ARHZ}, \cite{bass}, \cite{dFHWS97}, \cite{fcs}, \cite{dvor2}, \cite{he}, 
\cite{imke}, \cite{legall}, \cite{lyons}, \cite{sym}-\cite{yor2}, and $Z$ is
a normalization constant. Motivation for this construction came from polymer
physics (\textquotedblright excluded volume\textquotedblright\ effect),
while Symanzik \cite{sym} introduced the self-intersection local times as a
tool in constructive quantum field theory, see also \cite{dynk}.

A mathematically well-defined version of this ansatz was first given by
Varadhan \cite{varadhan} for $d=2$, and then by Westwater \cite{west} for $%
d=3$ .

"Stochastic quantization" addresses the - largely unresolved - challenge of
constructing random fields $\varphi $ whose probability measure obeys
certain physical postulates from quantum field theory. As introduced by
Parisi and Wu \cite{PW81}, this construction is attempted by introducing an
extra parameter $\tau $ and a stochastic differential equation with regard
to this parameter in such a way that for large $\tau $ the asymptotic
distribution of the Markov process $\varphi _{\tau }$ will satisfy those
postulates.

Conversely, for admissible measures $\mu $, local Dirichlet forms give rise
to such Markov processes with $\mu $ as their invariant measure. For the
2-dimensional Brownian motion Albeverio et.~al.~in \cite{ARHZ} have proven
the admissibility of the Edwards measure, properly renormalized as
elaborated by Varadhan \cite{varadhan}.

In this article we show in the framework of Dirichlet forms, that there
exists a Markov process which has the fractional Edwards measure as
invariant measure for the case that the Hurst parameter $H$ and the
dimension $d$ fulfill $Hd=1$. An analogous construction for  $Hd\leq 1$ 
can be found in \cite{GOSS} using
integration by parts techniques which are not available in this more
singular case. Instead the closability of the local pre-Dirichlet form
will be shown by quasi-translation-invariance w.r.t.\ shifts along the
Cameron-Martin space of fractional Brownian motion. 

In Section \ref{sec:framework} we shall introduce the required concepts and properties, so as
to then present our results and their proof in Section \ref{sec:main-results}.

\section{Preliminaries}
\label{sec:framework}
\subsection{Fractional Brownian Motion}
For $d\in\mathbb{N}$ and \emph{Hurst parameter} $H\in(0,1)$ a \emph{fractional
Brownian motion (fBm) in dimension $d$} is a $\mathbb{R}^{d}$-valued
centered Gaussian process $\big(B_{t}^{{\scriptscriptstyle {H}}}\big)_{t\ge0}$
with covariance, in case $d=1$: 
\begin{equation}
\label{cov-fBm}
\text{cov}_{{\scriptscriptstyle {H}}}(t,s):=\mathbb{E}\big[B_{t}^{{H}}B_{s}^{{H}}\big]=\frac{1}{2}\left(t^{2H}+s^{2H}-|t-s|^{2H}\right),\quad s,t\in[0,\infty).
\end{equation}
In $d$ dimensions we consider $d$ identical independent copies of
one-dimensional fractional Brownian motion.

In order to study the quasi translation invariance of the fractional Edwards measure (introduced below), we need to define the Cameron-Martin space associated to it.
The main role of the Cameron-Martin space is played by the fact that it characterizes precisely those directions in which translations leave the fractional Edwards measure "quasi-invariant" in the sense that the translated measure and the original measure have the same null sets. Here we give an abstract definition of the Cameron-Martin space for a Gaussian measure $\mu$ in a separable Banach space $B$ and later will realize it for the case at hand. The topological dual of the Banach space $B$ is denoted by $B^\prime$.

\begin{defn}[\cite{Hairer09}] \label{def_Cameron-Martin-space}
The Cameron-Martin space $K_\mu$ of a Gaussian measure $\mu$ on a separable reflexive Banach space $(B,\|\cdot\|)$ is the completion of the linear subspace $\tilde{K}_\mu\subset B$ defined by
\[
\tilde{K}_\mu:=\left\{h\in B\,|\,\exists h^*\in B^\prime\;\mathrm{with}\; \int_B h^*(x)\ell(x)\,d\mu(x)=\ell(h),\;\forall \ell\in B^\prime \right\}
\]
with respect to the norm $\| h\|^2_\mu:=\int_B | h^*(x) |^2\,d\mu(x)$. It becomes a Hilbert space when provided with the inner product 
$(h_1,h_2)_\mu:=\int_B h_1^*(x)h_2^*(x)\,d\mu(x)$.
\end{defn}

\begin{rem}
The norm $\|h\|_\mu$, hence the inner product $(h_1,h_2)_\mu$ in $\tilde{K}_\mu$, is well defined, that is they do not depend on the corresponding elements $h^*,h_1^*,h_2^*$ in $B'$, see Remark~3.26 in \cite{Hairer09}.
\end{rem}

To realize the fBm process let  $\Omega=X:=C_0([0,T],\mathbb{R}^{d})$ be the Banach space of all continuous paths in $\mathbb{R}^d$, null at time $0$, equipped with the supremum norm. 
Let
$\mathcal{B}_{H}$ denote the $\sigma$-algebra on $X$ generated by all maps
$X\ni\omega\mapsto B_{t}^{H}(\omega)\in \mathbb{R}^d$, $t\geq0$. The fractional Wiener measure 
on $X_H:=(X,\mathcal{B}_{H})$ we denote by $\nu_H$ and the expectation w.r.t.~$\nu_{{H}}$ is abbreviated by $\mathbb{E}_{H}(\cdot)$.
Let  $X_{H}'$ be the topological dual space of $X_H$ and $L^2:=L^{2}([0,T],\mathbb{R}^{d})$ 
the space of square integrable $\mathbb{R}^{d}$-valued functions on $[0,T]$.
Moreover let $\mathcal{H}_{H}$ be the Hilbert space defined by 
\[
\mathcal{H}_{H}=\overline{\{f\in L^{2}([0,T],\mathbb{R}^{d})\text{ such that }\|M_{H}f\|_{L^{2}}<\infty\}}.
\]
Here the operator $M_{H}$ is given by $(M_{H}f)(t)=\int_{0}^{T}\Lambda_{H}(t,s)f(s)\,ds$,
where $\Lambda_{H}$ is the fractional integral kernel, see \cite[eq.~(2.2)]{Decreusefond1998}
and \cite{Picard2011}. We denote by $\langle\cdot,\cdot\rangle_{H}=\langle M_{H}\cdot,M_{H}\cdot\rangle_{L^{2}}$
the inner product on $\mathcal{H}_{H}$. By identifying the Hilbert
space $\mathcal{H}_{H}$ with its dual we obtain the rigging $X_{H}\subset\mathcal{H}_{H}\subset X_{H}'$.
The dual pairing in a natural way generalizes the inner product on
$\mathcal{H}_{H}$.

Following \cite{NS09} a fractional version of the Cameron--Martin space $K_H$ of $\nu_H$ is hence given by
\begin{equation}
\label{CM-space}
K_H:=\left\{k:[0,T]\longrightarrow\mathbb{R}^d\,|\,\exists h\in L^2([0,T],\mathbb{R}^d),\; k_t=\int_0^tR_H(t,s)h(s)\,ds \right\},
\end{equation}
where $R_H$ is the square integrable kernel defined by
\[
R_H(t,s):=C_Hs^{\frac{1}{2}-H}\int_s^t(u-s)^{H-\frac{3}{2}}u^{H-\frac{1}{2}}\,du,\quad t>s,
\]
with $C_H=\sqrt{\frac{H(2H-1)}{\beta(2-2H,H-\frac{1}{2})}}$ and $\beta$ denotes the beta function. For $t\le s$ we put $R_H(t,s)=0$. The kernel $R_H$ is related to the covariance function of fBm in \eqref{cov-fBm} through the identity
\[
\mathrm{cov}_H(t,s)=\int_0^{t\wedge s}R_H(t,r)R_H(s,r)\,dr.
\] 

For a fBm $B^{H}=\{B_{t}^{H},\;t\ge0\}$ in $\mathbb{R}^{d}$ a shift along the
Cameron-Martin space $X^{H,u,k}$ is defined by 
\[
X^{H,u,k}:=\{X_{t}^{u,k}:=B_{t}^{H}+uk_{t},\;t\ge0\},\quad u\in\mathbb{R},\;k\in K_H.
\]
We use the notation
$$
\int \dot{k}\, d B^H
:= \int_0^T \dot{k}(s) \, d B^H_s,
$$
which is defined as in e.g.~\cite{Decreusefond1998}. Note that $\dot{k}$ is a well defined function in $L^2([0,T], \mathbb{R})$ due to \cite{NS09}.

\begin{lem}
\label{rem:CMspace} For a Gaussian measure $\nu$, in particular
for $\nu_{H}$, the shifted measure $\nu\circ\tau_{sk}$, where
$\tau_{sk}(\omega)=\omega+sk$, $s\in\mathbb{R}$ for $k$ from the
corresponding Cameron-Martin space $K_{H}$ is indeed quasi-translation
invariant, hence absolutely continuous w.r.t.~$\nu$, see e.g.~\cite{Kuo96}.
The Radon-Nikodym derivative, in the case of fractional Wiener measure $\nu_{H}$, 
is given by
\[
\frac{d\nu_{H}\circ\tau_{sk}}{d\nu_{H}}(B^{H})=\frac{1}{\mathbb{E}(\exp(s\langle dB^{H},\dot{k}\rangle_{H})}\exp(s\langle dB^{H},\dot{k}\rangle_{H}),\quad s\in\mathbb{R},
\]
where the first expression may be considered as an $L^{2}(\nu_{0})$
limit, $dB^{H}$ denotes the fractional white noise process and $\dot{k}$
the derivative of the function from the Cameron-Martin space $K_{H}$.
See also \cite{Picard2011}.
\end{lem}

\subsection{The Edwards Model}
The self-intersection local time of a fractional Brownian motion $B^{H}$
is given informally by 
\[
L(T):=L(T,B^{H}):=\int_{0}^{T}dt\int_{0}^{t}ds\,\delta(B_{t}^{H}-B_{s}^{H}).
\]
However it is well known that, for $Hd=1$ one has  $L(T)=\infty$ $\nu_{H}$-a.e.,
see e.g.~\cite{HN05}. Therefore a renormalization procedure is needed.
Let us use the heat kernel for the approximation of the $\delta$-function
\[
p_{\varepsilon}(x):=\frac{1}{(2\pi\varepsilon)^{\nicefrac{d}{2}}}e^{-\frac{|x|^{2}}{2\varepsilon}}=\frac{1}{(2\pi)^{d}}\int_{\mathbb{R}^{d}}e^{-\frac{\varepsilon}{2}|y|^{2}+i(y,x)}\,dy,\quad x\in\mathbb{R}^{d},
\]
which leads to the approximated self-intersection local time, see
also \cite{HN05} 
\begin{align*}
L_{\varepsilon}(T) & :=\int_{0}^{T}dt\int_{0}^{t}ds\,p_{\varepsilon}(B_{t}^{H}-B_{s}^{H})\\
 & =\frac{1}{(2\pi)^{d}}\int_{0}^{T}dt\int_{0}^{t}ds\int_{\mathbb{R}^{d}}e^{-\frac{\varepsilon}{2}|y|^{2}+i(y,B_{t}^{H}-B_{s}^{H})}\,dy.
\end{align*}
Moreover, as in \cite{varadhan} one has to center the local time
in order to perform the limit later on. Hence we define: 
\[
L_{\varepsilon,c}(T):=L_{\varepsilon}(T)-\mathbb{E}\big(L_{\varepsilon}(T)\big).
\]

In \cite{HN05} it is shown that for $\varepsilon\to0$ there exist
a limit of $L_{\varepsilon,c}(T)$ in the space of square integrable
functions. We denote: 
\[
L_{\varepsilon,c}(T)\to L_c=L_{c}(T),\quad\varepsilon\to0.
\]

In the case $Hd=1$, it is shown in \cite{GOSS} that, under certain
conditions on the coupling constant $g$, one has that the random variable
$e^{-gL_{c}}$ is a well defined object as an integrable function
w.r.t.~$\nu_{H}$. Hence we can define the fractional Edwards measure
in this case by 
\[
d\nu_{H,g}:=\frac{1}{\mathbb{E}(e^{-gL_{c}})}e^{-gL_{c}}\,d\nu_{H}.
\]
\begin{rem}
\begin{enumerate}
\item Note that by this definition $\nu_{H,g}$ is indeed a probability measure which is absolutely continuous
w.r.t.~the fractional Wiener measure $\nu_{H}$. We will hence use several times
that properties are holding $\nu_{H}$-a.e.~and hence $\nu_{H,g}$-a.e.
\item Notice also that the existence of the density as an $L^{1}(\nu_{H})$
function is not trivial due to the fact that, after centering the
random variable $L_{c}$ can indeed take negative values and the exponential
could become infinity. The ensurance of integrability, at least for
mild assumptions on $g$ is done in \cite{GOSS}.
\item The existence of certain exponential moments of $L_c$ was studied in \cite{HNS06}.
Due to this property the measure $\nu_{H,g}$ is also defined at least for some
negative $g$.
\end{enumerate}
\end{rem}

In the following we shall restrict our considerations to coupling constants $g$ such that $e^{-gL_{c}}\in L^{1}(\nu_{H})$, 
see \cite{GOSS}.

\subsection{Dirichlet Forms}
For the stochastic quantization we will use classical Dirichlet forms
of gradient type in the sense of \cite{AR91}. We start with a densely defined bilinear form of gradient type
$$
\mathbb{E}(f,g)= \int \langle \nabla f, \nabla g \rangle d m
$$ in a suitable $L^2(m)$ space and show closability. In many particular cases,
as in \cite{BFS16}, this can be done by an integration by parts argument.
Here however, due to the lack of Meyer-Watanabe differentiability
of the self-intersection local time for the case $Hd=1$, see e.g.~\cite{Hu2001} 
for the fBm case, the techniques are more involved.
Instead we show quasi-invariance of the fractional Edwards measure $\nu_{H,g}$
with respect to shifts in the Cameron-Martin space $K_H$ of fBm. 
Details on Dirichlet forms can be found in the monographs
\cite{BH91, FOT, MR92} and for the gradient Dirichlet forms, see \cite{AR91}.

As mentioned above we consider classical gradient Dirichlet forms, hence 
we have to introduce the gradient. To this end, at first we define the space of smooth cylinder functions. 
For a topological vector space $(\mathcal{X},\tau)$ we define the
set of smooth bounded cylinder functions
\[
\mathcal{F}C_{b}^{\infty}(\mathcal{X}):=\left\{ f(l_{1},\dots,l_{n})\,|\,n\in\mathbb{N},\,f\in C_{b}^{\infty}(\mathbb{R}^{n}),\,l_{1},\dots,l_{n}\in\mathcal{X}'\right\} ,
\]
where $C_{b}^{\infty}(\mathbb{R}^{n})$ is the space of bounded infinitely often differentiable
functions on $\mathbb{R}^{n}$, where all partial derivatives are
also bounded. 

For
$u\in\mathcal{F}C_{b}^{\infty}(X_H)$ and $\omega\in X_H$, following 
the notation \cite{AR91},  we define 
\[
\frac{\partial u}{\partial k}(\omega):=\frac{d}{ds}u(\omega+sk){\big|_{s=0}}.
\]
By $\nabla u(\omega)$ we denote the unique element in $\mathcal{H}_H$ such that
\[
\langle\nabla u(\omega),k\rangle_{H}=\frac{\partial u}{\partial k}(\omega),\quad\text{for all }k\in K_{H}.
\]


\begin{thm}
\label{thm closable} The bilinear form 
\[
\mathcal{E}_{H}(u,v):=\mathbb{E}_{H}(e^{-gL_c}\nabla u\cdot\nabla v),\quad u,v\in\mathcal{F}C_{b}^{\infty}(X_{H})
\]
is a symmetric pre-Dirichlet form, i.e., in particular closable, and gives
rise to a local, quasi-regular symmetric Dirichlet form in $L^{2}(X_{H},\nu_{{H,g}})$. 
\end{thm}

The proof of Theorem~\ref{thm closable} is given in Section~\ref{sec:main-results} which contains the proofs and main results. 
As indicated above we show closability of the bilinear form via quasi-translation invariance along shifts in the Cameron-Martin space $K_H$. 

\begin{rem}
	As in \cite[Cor.~10.8]{HKPS93} we obtain that the closures of $\big(\mathcal{E}_{H},\mathcal{F}C_b^\infty(X_H)\big)$ and $\big(\mathcal{E}_{H},\mathcal{P}\big)$ coincide, where $\mathcal{P}\subset L^{2}(X_{H},\nu_{{H,g}})$ denotes the dense subspace of polynomials. 
\end{rem}

\section{Main Results and Proofs}
\label{sec:main-results}
Crucial for the results of this paper is the following theorem. 

\begin{thm}
\label{thm_main} Let $k\in K_{H}$ be given and 
\[
a_{sk}:=\frac{d\nu_{H,g}\circ\tau_{sk}}{d\nu_{H,g}},\quad s\in\mathbb{R}.
\]
Then the process $(a_{sk})_{s\in\mathbb{R}}$ has a version which
has $\nu_{H}$-a.e. (and hence $\nu_{H,g}$-a.e.) continuous sample
paths. 
\end{thm}

We denote by $L(T,u,k)$ 
the self-intersection local time of $X^{H,u,k}$ and similarly for resp.\ $L_{\varepsilon}(T,u,k)$ and $L_{\varepsilon,c}(T,u,k)$ :
\begin{align*}
L(T,u,k) & :=\int_{0}^{T}dt\int_{0}^{t}ds\,\delta(X_{t}^{u,k}-X_{s}^{u,k}),\\
L_{\varepsilon}(T,u,k) & :=\frac{1}{(2\pi)^{d}}\int_{0}^{T}dt\int_{0}^{t}ds\int_{\mathbb{R}^{d}}dy\,e^{-\frac{\varepsilon}{2}|y|^{2}+i(y,X_{t}^{u,k}-X_{s}^{u,k})},\\
L_{\varepsilon,c}(T,u,k) & :=L_{\varepsilon}(T,u,k\big)-\mathbb{E}\big(L_{\varepsilon}\big(T\big)\big).
\end{align*}

To prove Theorem \ref{thm_main} we need the following two lemmata.
\begin{lem}
\label{lemma1} Let $\gamma\in(0,1)$ and $k\in K_{H}$ be given.
Then there exists a positive constant $C$ such that 
\begin{equation}
\|L_{\varepsilon,c}(T,u,k)-L_{\varepsilon,c}(T,v,k)\|_{L^{2}}^{2}\le C|u-v|^{1+\gamma}\label{eq:expectation1}
\end{equation}
for all $u,v\in\mathbb{R}$ and $\varepsilon>0$. 
\end{lem}

\begin{proof}
Explicitely  \eqref{eq:expectation1}.
\begin{align*}
 & L_{\varepsilon,c}\big(T,u,k)-L_{\varepsilon,c}\big(T,v,k)=L_{\varepsilon}(T,u,k\big)-L_{\varepsilon}(T,v,k\big)\\
= & \frac{1}{(2\pi)^{d}}\left(\int_{0}^{T}dt\int_{0}^{t}ds\int_{\mathbb{R}^{d}}dy\,\big(e^{-\frac{\varepsilon}{2}|y|^{2}}(e^{i(y,X_{t}^{u,k}-X_{s}^{u,k})}-e^{i(y,X_{t}^{v,k}-X_{s}^{v,k})}\big)\,dy\right)\\
= & \frac{1}{(2\pi)^{d}}\left(\int_{0}^{T}dt\int_{0}^{t}ds\int_{\mathbb{R}^{d}}dy\,e^{-\frac{\varepsilon}{2}|y|^{2}}\big(e^{iu(y,k_{t}-k_{s})}-e^{iv(y,k_{t}-k_{s})}\big)e^{i(y,B_{t}^{H}-B_{s}^{H})}\,dy\right)
\end{align*}
which implies 
\begin{align*}
 & \big|L_{\varepsilon,c}\big(T,u,k)-L_{\varepsilon,c}\big(T,v,k)\big|^{2}\\
 & =\frac{1}{(2\pi)^{2d}}\int_{\mathcal{T}}d\tau\int_{\mathbb{R}^{2d}}dy\,e^{-\frac{\varepsilon}{2}(|y_{1}|^{2}+|y_{2}|^{2})}\big(e^{iu(y_{1},k_{t}-k_{s})}-e^{iv(y_{1},k_{t}-k_{s})}\big)\\
 & \times\big(e^{-iu(y_{2},k_{t'}-k_{s'})}-e^{-iv(y_{2},k_{t'}-k_{s'})}\big)e^{i(y_{1},B_{t}^{H}-B_{s}^{H})-i(y_{2},B_{t'}^{H}-B_{s'}^{H})},
\end{align*}
where $d\tau=dsdtds'dt'$, $dy=dy_{1}dy_{2}$ and 
\[
\mathcal{T}:=\{(s,t,s',t'):0<s<t<T,\;0<s'<t'<T\}.
\]
Computing the expectation 
\begin{equation}
\mathbb{E}\big(e^{i(y_{1},B_{t}^{H}-B_{s}^{H})-i(y_{2},B_{t'}^{H}-B_{s'}^{H})}\big)=\prod_{j=1}^{d}\mathbb{E}\big(e^{i(y_{1j},B_{t}^{H,j}-B_{s}^{H,j})-i(y_{2j},B_{t'}^{H,j}-B_{s'}^{H,j})}\big)=e^{-\frac{1}{2}(y,\Sigma y)},\label{eq:expectation2}
\end{equation}
where $y=\left(\begin{smallmatrix}y_{1}\\
y_{2}
\end{smallmatrix}\right)$ and $\Sigma=\left(\begin{smallmatrix}\lambda & \mu\\
\mu & \rho
\end{smallmatrix}\right)$ is a symmetric matrix with 
\begin{align*}
\lambda & =|t-s|^{2H},\\
\rho & =|t'-s'|^{2H},\\
\mu & =\frac{1}{2}\left(|t-s'|^{2H}+|t'-s|^{2H}-|t-t'|^{2H}-|s-s'|^{2H}\right).
\end{align*}
Thus, the lhs of \eqref{eq:expectation1} is equal to 
\begin{multline*}
\|L_{\varepsilon,c}(T,u,k)-L_{\varepsilon,c}(T,v,k)\|_{L^{2}}^{2}\\
=\frac{1}{(2\pi)^{2d}}\int_{\mathcal{T}}d\tau\int_{\mathbb{R}^{2d}}dy\,e^{-\frac{1}{2}\left((y,\Sigma y)+\varepsilon|y|^{2}\right)}\\
\times\big(e^{iu(y_{1},k_{t}-k_{s})}-e^{iv(y_{1},k_{t}-k_{s})}\big)\big(e^{-iu(y_{2},k_{t'}-k_{s'})}-e^{-iv(y_{2},k_{t'}-k_{s'})}\big).
\end{multline*}
Notice that for any given $\alpha\in(0,1]$ there exists a constant
$C\in(0,\infty)$ (from now on, the constant $C$ might be different
from line to line) such that 
\begin{align*}
|\cos(x)-\cos(y)| & \leq C|x-y|^{\alpha}\wedge1,\\
|\sin(x)-\sin(y)| & \leq C|x-y|^{\alpha}\wedge1.
\end{align*}
On the other hand, we have 
\begin{align*}
 & \big(e^{iu(y_{1},k_{t}-k_{s})}-e^{iv(y_{1},k_{t}-k_{s})}\big)\big(e^{-iu(y_{2},k_{t'}-k_{s'})}-e^{-iv(y_{2},k_{t'}-k_{s'})}\big)\\
 & =\big(\cos(u(y_{1},k_{t}-k_{s}))-\cos(v(y_{1},k_{t}-k_{s}))\big)\big(\cos(u(y_{2},k_{t'}-k_{s'}))-\cos(v(y,k_{t'}-k_{s'}))\big)\\
 & -\big(\sin(u(y_{1},k_{t}-k_{s}))-\sin(v(y_{1},k_{t}-k_{s}))\big)\big(\sin(v(y_{2},k_{t'}-k_{s'}))-\sin(u(y_{2},k_{t'}-k_{s'}))\big)\\
 & +i(\text{cross\;terms}),
\end{align*}
where the ``cross terms'' are odd functions. Hence the $y_{1},y_{2}$-integral
with these functions vanishes. Finally we obtain the following estimate
for the lhs of \eqref{eq:expectation1} . 
\begin{multline*}
\|L_{\varepsilon,c}(T,u,k)-L_{\varepsilon,c}(T,v,k)\|_{L^{2}}^{2}\\
\leq C|u-v|^{2\alpha}\int_{\mathcal{T}}d\tau\int_{\mathbb{R}^{2d}}dy\,e^{-\frac{1}{2}\left((y,\Sigma y)+\varepsilon|y|^{2}\right)}(|y_{1}||y_{2}|)^{2\alpha}.
\end{multline*}
Here we used the fact that functions from the Cameron-Martin space
are continuous since they are given as fractional integral operators acting on square
integrable functions, which allows to bound them in supremum norm, see \cite{Picard2011}.

If we denote by $I_{d}$ the $d\times d$ identity matrix, then the
Gaussian integral is equal to 
\[
\int_{\mathbb{R}^{2d}}dy\,e^{-\frac{1}{2}\left((y,\Sigma y)+\varepsilon|y|^{2}\right)}(|y_{1}||y_{2}|)^{2\alpha}=2^{d(2\alpha+1)}\Gamma\left(\alpha+\frac{1}{2}\right)^{2d}\frac{1}{\det(\Sigma+\varepsilon I_{d})^{\frac{d}{2}+d\alpha}},
\]
Summarizing we obtain 
\[
\|L_{\varepsilon,c}(T,u,k)-L_{\varepsilon,c}(T,v,k)\|_{L^{2}}^{2}\le C|u-v|^{2\alpha}\int_{\mathcal{T}}d\tau\frac{1}{[(\lambda+\varepsilon)(\rho+\varepsilon)-2\mu]^{\frac{d}{2}+d\alpha}}<\infty,
\]
by Lemma 11 in \cite{HN05} and the fact that for every $\varepsilon>0$
the above integral has no singularities. Taking $\alpha\in\left[\frac{1}{2},1\right)$ yields the desired statement. 
\end{proof}
\begin{lem}
\label{lemma2} Let $Y(u,k):=L_{c}(B^{H}+uk)-L_{c}(B^{H}),$ for $u\in\mathbb{R}$
and $k\in K_{H}$. Then for any $\gamma\in(0,1)$, there exists a
constant $0<C<\infty$ such that 
\[
\mathbb{E}\|Y(u,k)-Y(v,k)\|^{2}\leq C|u-v|^{1+\gamma},\quad u,v\in\mathbb{R}.
\]
\end{lem}
\begin{proof} 
We have from \cite{HN05} that $L_{\varepsilon,c}$ is convergent
in $L^{2}(\nu_{H})$ for $\varepsilon\to0$. Hence there is a sequence
$\varepsilon_{n}\to0$ such that $L_{\varepsilon_{n},c}\to L_{c}$
in probability w.r.t.~$\nu_{H}$. Hence $L_{\varepsilon_{n},c}(\cdot+uk)\to L_{c}(\cdot+uk)$
in probability w.r.t.~$\nu_{H}$. Therefore $L_{\varepsilon_{n},c}(\cdot+uk)-L_{\varepsilon_{n},c}(\cdot)\to Y(u,k)$
in probability w.r.t.~$\nu_{H}$. This gives immediately the
desired result by Lemma \ref{lemma1}. 
\end{proof}
Now we have all ingredients to prove the Theorem \ref{thm_main}.
\begin{proof}[Proof of Theorem \ref{thm_main}]
By Lemma \ref{lemma2} we know that for any $k\in K_{H}$  and $u\in\mathbb{R}$ there is
a version $\tilde{Y}(u,k)$, i.e~
$$
\nu_{H}\left(Y(u,k)=\tilde{Y}(u,k)\right)=1,\quad u\in\mathbb{R},
$$
such that
$$
\nu_{H}\left(\tilde{Y}(u,k)\text{ is continuous with respect to }u\in\mathbb{R}\right)=1.
$$
By definition of the fractional Edwards measure 
\[
\nu_{H,g}=\frac{1}{\mathbb{E}(e^{-gL_{c}})}e^{-gL_{c}}\nu_{H},
\]
it is clear that $\nu_{H,g}\circ\tau_{uk}$ is absolutely continuous
w.r.t.~$\nu_{H,g}$ for all $u\in\mathbb{R}$ and $k\in K_{H}$.
Then by Lemma \ref{rem:CMspace} we know that 
\[
a_{uk}=e^{-uY(u,k)}\frac{\exp(u\int \dot{k}\, d B^H)}{\mathbb{E}\big(\exp(u\int \dot{k}\, d B^H)\big)}.
\]
Now let

\[
\tilde{a}_{uk}=e^{-u\tilde{Y}(u,k)}\frac{\exp(u\int \dot{k}\, d B^H)}{\mathbb{E}\big(\exp(u\int \dot{k}\, d B^H)\big)}.
\]
Hence we have, with the previous consideration of $\tilde{Y}(u,k)$
\begin{eqnarray*}
\nu_{H}\left(\tilde{a}_{uk}\text{ is continuous with respect to }u\in\mathbb{R}\right)&=&1,
\end{eqnarray*}
and due to the absolute continuity of $\nu_{H,g}$ w.r.t.~$\nu_{H}$
the same holds for $\nu_{H,g}$ which shows the assertion. 
\end{proof}

\begin{proof}[Proof of Theorem~\ref{thm closable}]
Since the Cameron-Martin space $K_{H}$ is dense in $\mathcal{H}_{H}$ we can find an orthonormal
basis $(k_{n})_{n}$ such that the bilinear form on $\mathcal{F}C_{b}^{\infty}(X_{H})$
can be written as 
\[
\mathcal{E}_{H}(u,v)=\sum_{n=1}^{\infty}\int\frac{\partial u}{\partial k_{n}}\frac{\partial v}{\partial k_{n}}d\nu_{H,g}.
\]
From Proposition 3.7 in \cite{MR92} Chapter I it suffices to show closability
for every $n$ separately. However this is a direct consequence of
Theorem \ref{thm_main} and Corollary 2.5 in \cite{AR91}. Hence as in the proof of Proposition 3.5 in \cite{MR92} Chapter II, Section 3a)
we obtain a Dirichlet form as the closure $\big(\mathcal{E}_{H},D(\mathcal{E}_{H})\big)$
of the above quadratic form. For locality, see Example 1.12(ii) in \cite{MR92} Chapter V and for quasi-regularity, see \cite{MR92} Chapter IV Section 4b).
\end{proof}
 As a direct consequence of Theorem 3.5 in \cite{MR92} Chapter IV and Theorem 1.11 in \cite{MR92} Chapter V we have:
\begin{thm}
\label{thm diffusion} There exists a diffusion process 
\[
\mathbb{M}_{H}=(\Omega,\mathcal{F},(\mathcal{F}_{t})_{t\geq0},(X_{t})_{t\geq0},(P_{\omega})_{\omega\in X_{H}})
\]
with state space $X_H$ which is properly associated with $(\mathcal{E}_{H},\mathcal{D}(\mathcal{E}_{H}))$. In particular, $\mathbb{M}_{H}$ is $\nu_{H,g}$-symmetric and has $\nu_{H,g}$ as invariant measure. 
\end{thm}

\section*{Conclusion}
In this work we showed the existence of a Markov process having the fractional Edwards measure for $Hd=1$ as an invariant measure. The process is obtained as a Hunt process associated to the symmetric Dirichlet form $\mathcal{E}_H$. Closability of the form was shown using quasi-translation invariance of the  fractional Edwards measure w.r.t.~shifts along the Cameron-Martin space. This generalizes the results found in \cite{BFS16} for the case $Hd<1$, where the closability was proved by integration by parts. This is not possible in the present case ($Hd=1$) due to the lack of Meyer-Watanabe differentiability of the density. 
The explicit representation of the generator is known in the case $Hd<1$ by standard integration by parts techniques, see \cite{BFS18}. In the case $Hd=1$ this however is unknown. To characterize the Markov process in the present situation we plan to use Mosco convergence in the Hurst parameter $H$ for approximating Dirichlet forms and hence to obtain convergence of the associated operator semigroups. 

\section*{Acknowledgement}

We truly thank M.~J.~Oliveira and R.~V.~Mendes for the hospitality
at the Academy of Science in Lisbon. We thank G.~Trutnau for helpful
discussions. Financial support by the Mathematics Department of the
University of Kaiserslautern for research visits at Lisbon and project
UID/MATH/04674/2013 are gratefully acknowledged.

\bibliographystyle{plain} 

\end{document}